\documentclass[pdflatex,sn-mathphys-num]{sn-jnl}

\usepackage{graphicx}%
\usepackage{multirow}%
\usepackage{amsmath,amssymb,amsfonts}%
\usepackage{amsthm}%
\usepackage{mathrsfs}%
\usepackage[title]{appendix}%
\usepackage{xcolor}%
\usepackage{textcomp}%
\usepackage{manyfoot}%
\usepackage{booktabs}%
\usepackage{algorithm}%
\usepackage{algorithmicx}%
\usepackage{algpseudocode}%
\usepackage{listings}%

\usepackage{enumitem}
\setlist[description]{font=\normalfont\bfseries} 
\usepackage{caption}
\theoremstyle{thmstyleone}%
\newtheorem{theorem}{Theorem}
\newtheorem{proposition}[theorem]{Proposition}%

\theoremstyle{thmstyletwo}%
\newtheorem{remark}{Remark}%

\theoremstyle{thmstylethree}%

\begin{document}

\title[Article Title]{Quantum Meet-in-the-Middle Attacks on Key-Length Extension Constructions}

\author*[1]{\fnm{Min} \sur{Liang}}\email{liangmin07@mails.ucas.ac.cn}
\author[1,2]{\fnm{Ruihao} \sur{Gao}}
\author[1]{\fnm{Jiali} \sur{Wu}}

\affil*[1]{\orgname{Data Communication Science and Technology Research Institute}, \orgaddress{\city{Beijing}, \postcode{100191}, \country{China}}}
\affil[2]{\orgname{Graduate School of China Academy of Telecommunication Technology}, \orgaddress{\city{Beijing}, \postcode{100191}, \country{China}}}

\abstract{Key-length extension (KLE) techniques provide a general approach to enhancing the security of block ciphers by using longer keys. There are mainly two classes of KLE techniques, cascade encryption and XOR-cascade encryption. This paper presents several quantum meet-in-the-middle (MITM) attacks against two specific KLE constructions.

For the two-key triple encryption (2kTE), we propose two quantum MITM attacks under the Q2 model. The first attack, leveraging the quantum claw-finding (QCF) algorithm, achieves a time complexity of $O(2^{2\kappa/3})$ with $O(2^{2\kappa/3})$ quantum random access memory (QRAM). The second attack, based on Grover's algorithm, achieves a time complexity of $O(2^{\kappa/2})$ with $O(2^\kappa)$ QRAM. The latter complexity is nearly identical to Grover-based brute-force attack on the underlying block cipher, indicating that 2kTE does not enhance security under the Q2 model when sufficient QRAM resources are available.

For the 3XOR-cascade encryption (3XCE), we propose a quantum MITM attack applicable to the Q1 model. This attack requires no QRAM and has a time complexity of $O(2^{(\kappa+n)/2})$ ($\kappa$ and $n$ are the key length and block length of the underlying block cipher, respectively.), achieving a quadratic speedup over classical MITM attack.

Furthermore, we extend the quantum MITM attack to quantum sieve-in-the-middle (SITM) attack, which is applicable for more constructions. We present a general quantum SITM framework for the construction $ELE=E^2\circ L\circ E^1$ and provide specific attack schemes for three different forms of the middle layer $L$. The quantum SITM attack technique can be further applied to a broader range of quantum cryptanalysis scenarios.}



\maketitle

\section{Introduction}
When a block cipher has a small key size and fails to provide the required security level, key-length extension (KLE) techniques offer a general approach to enhancing security by employing longer keys without modifying the underlying block cipher. A canonical example is 3DES, which utilizes a triple encryption cascade with longer keys to augment the security of DES.
Currently, the researchers had proposed various KLE techniques, including double/triple encryption and FX construction, and studied their security from provable security and cryptanalysis \cite{dinur2012,gazi2013,lee2013,gazi2015,merkle1981,kilian2001}.

While classical cryptanalysis has established the security of KLE techniques such as triple encryption, the emergence of quantum algorithms presents new challenges, notably the quadratic speedups enabled by Grover's algorithm. Furthermore, quantum algorithms can accelerate the resolution of numerous computational problems and intensify attacks on certain cryptographic constructions \cite{bonnetain2022,leander2017,bonnetain2019}. This acceleration may lead to the failure of some KLE techniques in classical model to enhance security within the quantum model. Therefore, the security of KLE techniques in the quantum setting demands careful attention.

This article investigates the quantum security of two KLE techniques through cryptanalytic methods, aiming to provide a foundation for evaluating the security enhancement efficacy of KLE techniques in the quantum model.

\noindent\textbf{Cascade and XOR-cascade encryptions.} These are two fundamental classes of KLE designs.

Cascade encryption involves using a block cipher $l(l\geq 2)$ times in a cascading fashion, denoted as $CE^l (M)=E_{k_l}(\cdots E_{k_2}(E_{k_1}(M))\cdots)$. For $l=2$ and $l=3$, $CE^l$ corresponds to double and triple encryption, respectively. A notable example is 3DES, a triple encryption variant where the second encryption is replaced by decryption.

XOR-cascade encryption consists of an $l$-layer cascade interleaved with $l+1$ whitening keys, formalized as $XCE^l (M)=E_{k_l}(\cdots E_{k_2}(E_{k_1}(M\oplus z_1)\oplus z_2)\cdots \oplus z_{k_l})\oplus z_{k_{l+1}}$. When $l=1$, $XCE^l$ reduces to the FX construction \cite{kilian2001}, which has been extensively studied in both classical and quantum contexts. For $l=2$, the 2XOR and 3XOR constructions have been primarily analyzed in quantum settings \cite{bonnetain2022}.

\noindent\textbf{Cryptanalysis model.} When analyzing cipher security, the definition of a cryptanalysis model is essential. For KLE techniques, this analysis involves two key dimensions.
\begin{description}
  \item[(1) Adversarial Capabilities.] Three primary models are considered \cite{zhandry2012,liang2021}.
        \begin{itemize}
          \item Classical Model. Attackers are restricted to online classical queries and offline classical computations.
          \item Quantum Q1 Model. Attackers can perform offline quantum computations but are limited to online classical queries.
          \item Quantum Q2 Model. Attackers have full quantum capabilities, including local quantum computations and superposition queries to remote quantum oracles.
        \end{itemize}
        The classical model imposes the weakest constraints, while Q2 grants the strongest adversarial power. Our study evaluates security in Q1 and Q2 models, comparing results against the classical baseline.
  \item[(2) Complexity Metrics.] There are two standard models for measuring attack complexity \cite{dinur2016}.
        \begin{itemize}
          \item Information-Theoretic Model. Complexity is quantified by the total number of queries to the block cipher and KLE construction, ignoring real computational costs.
          \item Computational Model. Complexity is measured by operational counts, often including memory requirements.
        \end{itemize}
        Prior work \cite{gazi2013,lee2013,gazi2015} analyzed $CE^l$ and $XCE^l$ in the information-theoretic model, showing query complexity approaching $2^{\kappa+n}$ as $l$ increases (where $\kappa$ and $n$ denote key and block lengths, respectively). This work focuses exclusively on the computational model for quantum attacks.
\end{description}

\noindent\textbf{Security of cascade encryption.} The meet-in-the-middle (MITM) attack remains a pivotal approach in cryptanalyzing cascade structures, with prior research spanning classical and quantum frameworks.

For double encryption $CE^2 (M)=E_{k_2}(E_{k_1}(M))$, $k_1,k_2\in\{0,1\}^\kappa$, $M\in\{0,1\}^n$, classical MITM attacks exhibit $O(2^\kappa)$ complexity \cite{diffie1977}, matching the underlying block cipher's brute-force bound. This confirms $CE^2$ provides no asymptotic security gain in classical model\footnote{Ref.\cite{jaeger2021} proved that double encryption can enhance the quantum security in Q2 model.}, necessitating $l\geq 3$ for $CE^l$ in classical settings.

A classical MITM extension \cite{merkle1981} cracks $CE^3$ in $2^{2\kappa}$ time using $2^\kappa$ memory. In Q1 model, Zhong and Bao's quantum MITM attack \cite{zhong2010} targets 3DES and general $CE^3$, achieving $O(2^\kappa)$ time/QRAM. For the two-key variant ($k_3=k_1$), its complexity aligns with Grover-based brute-force attack($O(2^\kappa)$), so their quantum attack fails for the two-key variant.

For $CE^r (r\geq 4)$, MITM attack is outperformed significantly by dissection attack. By classical dissection attack \cite{dinur2012}, $CE^4$ can be broken in time $O(2^{n+\kappa})$ with $O(2^\kappa)$ classical memory. Kaplan \cite{kaplan2014} proposed quantum dissection attack in Q1 model, which further reduces this to $O(2^{n/2+2\kappa/3})$ time and $O(2^{2\kappa/3})$ QRAM.

We focus on two-key triple encryption (2kTE), and propose new quantum MITM attacks that would be better than Grover-based brute-force attack.

\noindent\textbf{Security of XOR-cascade encryption.} Prior studies on $XCE^l$ in the information-theoretic model are detailed in \cite{gazi2013,lee2013,gazi2015}, while this work focuses exclusively on complexity analysis in the computational model.

As a pivotal KLE technique, the FX construction has been studied in both classical and quantum contexts. Leander and May \cite{leander2017} introduced a quantum attack combining Grover's and Simon's algorithms, later optimized by Bonnetain et al. \cite{bonnetain2019} to break FX in $O(n^3 2^{\kappa/2})$ time under Q2 model. This demonstrates FX offers no substantial security gain over the underlying block cipher. When adapted to Q1 model, the attack satisfies $T^2 D=O(2^{n+\kappa})$ with $D\leq 2^n$ ($T$ and $D$ are the time complexity and data complexity, respectively.), matching the optimal security bound in Q1 model \cite{jaeger2021,guo2024}.

Proposed by Ga\v{z}i and Tessaro\cite{gazi2012}, 2XOR features double block encryption with two XOR operations, defined as $2XOR_{k,z}(M)=E_{\pi(k)}(E_k (M\oplus z)\oplus z)$ ($\pi$ is a public fixpoint-free permutation). Bonnetain et al. \cite{bonnetain2022} showed that 2XOR's quantum attack complexity in Q1/Q2 model is the same as that of FX, leveraging similar techniques from \cite{bonnetain2019}. Notably, appending a third whitening key (yielding $2XOR_{k,z}(M)\oplus z$) may invalidate this attack, prompting the variant denoted as 3XOR in \cite{gazi2015} due to three XOR operations.

As a specific case of $XCE^2$ with three XOR operations\footnote{Both 2XOR and 3XOR are the special instances of $XCE^2$, where the third whitening key is set to be zero or not.}, 3XOR-cascade encryption (3XCE) is the focus of our quantum cryptanalysis, for which we present novel cryptanalytic techniques and complexity bounds.

\noindent\textbf{Organization of the paper.} Section \ref{sec2} introduces fundamental notations and reviews Grover's quantum search algorithm and the quantum claw-finding algorithm. Section \ref{sec3} focuses on the quantum security of 2kTE, presenting two quantum MITM attacks under the Q2 model. Section \ref{sec4} applies the Q2-model attacks from Section \ref{sec3} to analyze 3XCE, then devises a novel quantum MITM attack for 3XCE in the Q1 model. Section \ref{sec5} generalizes the quantum MITM framework to a quantum sieve-in-the-middle attack.

\section{Preliminaries}\label{sec2}
\subsection{Notations}
Denote $[t]$ as the set $\{i\in\mathbb{N}|1\leq i\leq t\}$. Denote $|X|$ as the size of the set $X$. Given two sets $X,Y$, let $X\times Y=\{(x,y)|x\in X,y\in Y\}$. $x\in X\backslash Y$ represents $x\in X$ and $x\notin Y$. Define ``$\parallel$" as the concatenation of two bit-strings, then $x\parallel y\in\{0,1\}^{2n}$ if $x,y\in\{0,1\}^n$. $\oplus_k(\cdot)$ represents a function which maps the input $x$ to $x\oplus k$. Given two functions $F,P$ with matching domain and range, denote $F\circ P$ as the composition of $F$ and $P$.

A block cipher is a family of functions $E:\{0,1\}^\kappa \times \{0,1\}^n\rightarrow \{0,1\}^n$, where $E(k,\cdot)$ is a permutation on $\{0,1\}^n$ for any given key $k\in\{0,1\}^\kappa$. Denote $BC(\kappa,n)$ as the set of all block ciphers with $\kappa$-bit key and $n$-bit block. For simplify our description, $E(k,m)$ is also written as $E_k (m)$ somewhere, and its corresponding decryption is denoted as $D(k,c)$ or $D_k (c)$ or $E_k^{-1}(c)$.

Given a function $f:X\rightarrow Y$, its quantum Oracle is denoted as $O_f$, which implements quantum transformation $O_f:|x\rangle|z\rangle\rightarrow|x\rangle|z\oplus f(x)\rangle$ or $O_f:|x\rangle\rightarrow (-1)^{f(x)}|x\rangle$.

\subsection{Grover's Quantum Search Algorithm}
In quantum cryptanalysis, Grover's algorithm \cite{grover1996} is one of the mostly used tools for analyzing quantum security.
Given an unsorted space of size $2^n$, we want to search one element $x$ such that $f(x)=1$. If there are $M$ elements that can satisfy $f(x)=1$, $O(2^n/M)$ queries to $f$ are necessary on average while classical algorithm is used. However, Grover's algorithm can solve it using only $O(\sqrt{2^n/M})$ queries to $f$. That means Grover's algorithm can solve the problem with quadratic speedup than classical algorithm. Moreover, Grover's algorithm has been proved in theory to achieve the optimal complexity bound \cite{zalka1999}.

Grover's algorithm can be described as three steps.
\begin{description}
  \item[(1)] Prepare the initial quantum state $|s\rangle=\frac{1}{2^{n/2}}\sum_{x\in\{0,1\}^n}|x\rangle$;
  \item[(2)] Repeat quantum oracle $O_f$ and unitary transformation $U_0$ for $\lceil\frac{\pi}{4}\sqrt{2^n/M}\rceil$ times, where
                \begin{eqnarray*}
                &&O_f:|x\rangle\rightarrow (-1)^{f(x)}|x\rangle,\forall x\in\{0,1\}^n,\\
                &&U_0=H^{\otimes n}(2|0^n\rangle\langle 0^n|-I)H^{\otimes n}.
                \end{eqnarray*}
  \item[(3)] Perform quantum measurement and obtain $x'\in\{x\in\{0,1\}^n |f(x)=1\}$ with high probability $O(1)$.
\end{description}

In Grover's algorithm, the quantum blackbox $O_f$ is the only component dependent on the specific problem. When applying the algorithm to cryptanalysis, a specific function $f$ must be defined based on both the quantum attack strategy and the cryptographic construction. Crucially, it is essential to ensure that the oracle $O_f$ remains accessible to an attacker without knowledge of the secret key.

\subsection{Quantum Claw-Finding Algorithm}\label{sec23}
Let us consider two functions $f:X\rightarrow Z$ and $g:Y\rightarrow Z$. If there exists a pair $(x,y)\in X\times Y$ such that $f(x)=g(y)$, then it is called a claw of the functions $f$ and $g$. Now we define the claw-finding problem as follows.

\noindent\textbf{Problem 1 (Claw-finding problem).} Assume the two functions $f:X\rightarrow Z$ and $g:Y\rightarrow Z$ are given as oracles, find a claw of $f$ and $g$.

In this article ,we adapt the quantum claw-finding (QCF) algorithm based on quantum random walks, which is referred to \cite{zhang2005,tani2007}. If $\sqrt{|X| }\leq |Y|<|X|^2$, QCF algorithm can solve claw-finding problem in time $O((|X|\cdot|Y|)^{1/3})$ using QRAM $O((|X|\cdot|Y|)^{1/3})$; If $|Y|\geq |X|^2$, QCF algorithm can solve claw-finding problem in time $O(|Y|^{1/2})$ using QRAM $O(|Y|^{1/2})$.

Numerous cryptanalysis problems can be reduced to the claw-finding problem, enabling the QCF algorithm to be applied for cipher attacks \cite{kaplan2014,jaques2019}. In this work, the QCF algorithm is employed in the quantum MITM attack. As mandated by the QCF algorithm, two specific functions $f$ and $g$ must be defined based on the cryptographic construction, with the assurance that the oracles $O_f$ and $O_g$ are accessible to an attacker lacking knowledge of the secret key.

\section{Two-Key Triple Encryption}\label{sec3}
The 2kTE scheme (see Fig. \ref{fig1}) can be described as follow.
\begin{equation}
2kTE_{k_1,k_2}(m)=E_{k_1}(E_{k_2}(E_{k_1}(m))),
\end{equation}
where $E\in BC(\kappa,n)$ is underlying block cipher, $k_1,k_2\in\{0,1\}^\kappa$ are two independent keys. The underlying block cipher $E$ is called three times, where the key $k_1$ is used in the first and third calls, and the key $k_2$ is used only in the second call.
\begin{figure}[h]
\centering
  \includegraphics[scale=0.23]{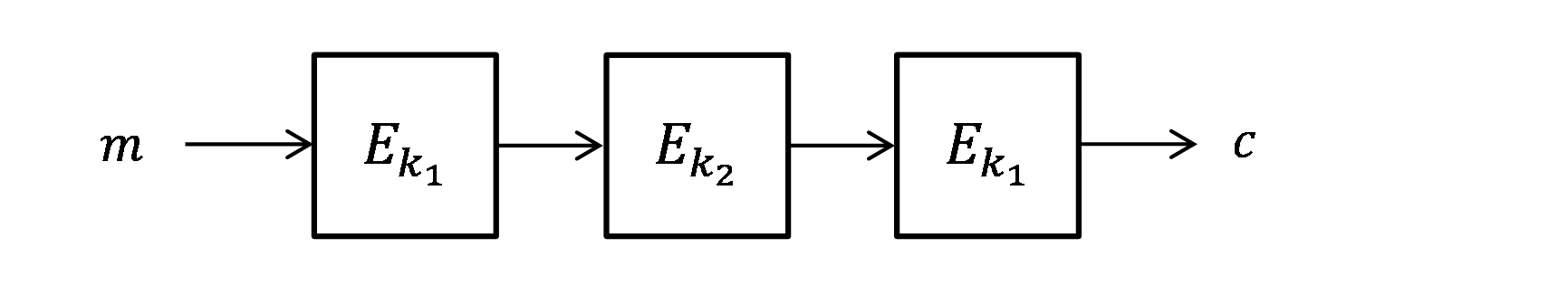}
\caption{Two-key triple encryption scheme.}
\label{fig1}       
\end{figure}

In this section, we present two types of quantum MITM attacks against the 2kTE scheme under the Q2 model, leveraging the QCF algorithm and Grover's algorithm, respectively.

\subsection{Quantum Attack Based on Quantum Claw-Finding Algorithm}\label{sec31}
Our first quantum MITM attack against 2kTE is devised by leveraging the QCF algorithm. The core strategy involves reducing the MITM attack on 2kTE to a claw-finding problem, which is then solved using quantum techniques.

We first define two functions $f,g:\{0,1\}^\kappa \rightarrow \{0,1\}^{tn}$ as follows.
\begin{eqnarray}
f(x)=&&\left(E_x^{-1}\left(2kTE_{k_1,k_2}\left(E_x^{-1}(1)\right)\right)\parallel E_x^{-1}\left(2kTE_{k_1,k_2}\left(E_x^{-1}(2)\right)\right)\parallel\cdots\right.\nonumber\\
 &&\left.\cdots\parallel E_x^{-1}\left(2kTE_{k_1,k_2}\left(E_x^{-1}(t)\right)\right)\right),\\
g(y)=&&\left(E_y(1)\parallel E_y(2)\parallel\cdots\parallel E_y(t)\right),
\end{eqnarray}
where $t=\lceil 2\kappa/n\rceil+1$, and $(k_1,k_2)$ denote the true keys of the 2kTE scheme.

Next we prove there exists a unique claw $(k_1,k_2)$ of $f,g$ with high probability. This assures that the claw found by QCF algorithm is the true key with high probability.

\begin{proposition}\label{prop1}
The true keys $(k_1,k_2)$ of 2kTE scheme must be a claw such that $f(k_1)=g(k_2)$.
\end{proposition}
\begin{proof}
According to the definitions of $2kTE_{k_1,k_2}$ and $f,g$, it can be easily induced as follow. If $x=k_1$,
$$E_x^{-1}(2kTE_{k_1,k_2}(E_x^{-1}(i)))=E_x^{-1} \left(E_{k_1}\left(E_{k_2}\left(E_{k_1}\left(E_x^{-1} (i)\right)\right)\right)\right)=E_{k_2}(i),\forall i\in[t].$$
Thus we have $f(k_1)=g(k_2)$.
\end{proof}

\begin{proposition}\label{prop2}
If $t=\lceil 2\kappa/n\rceil+1$, there exists a claw $(x,y)$ ($(x,y)\neq (k_1,k_2)$) of $f$ and $g$ with probability at most $2^{-n+1}$.
\end{proposition}
\begin{proof}
For an ideal block cipher $E:\{0,1\}^\kappa\times\{0,1\}^n\rightarrow\{0,1\}^n$, if the plaintext $i\in[t]$ is fixed and the key $y\in\{0,1\}^\kappa$ is an unknown variable, the ciphertext $E_y (i)$ can be treated as a random value in the space $\{0,1\}^n$, with all values $E_y (i),\forall i\in[t]$ being independent. Consequently, $g(y)$ constitutes a random value in $\{0,1\}^{tn}$.

Given the true key $(k_1,k_2)$, $2kTE_{k_1,k_2}(\cdot)$ constitutes a fixed permutation. By the definition of $f(x)$, $E_x^{-1}(i)$ for all $i\in[t]$ can be treated as a random value, implying that $f(x)$ itself is also a random value. As $x$ iterates over all elements in the space $\{0,1\}^\kappa$, computing $f(x)$ yields approximately $2^\kappa$ distinct values of length $tn$. Let $S_f=\{f(x)\in\{0,1\}^{tn}|x\in\{0,1\}^\kappa\}$; thus, $|S_f |\leq 2^\kappa$. From the definition of $S_f$, the existence of a claw for $f$ and $g$ is equivalent to the existence of some $y$ such that $g(y)\in S_f$.

By Proposition \ref{prop1}, $g(k_2)$ is necessarily an element of $S_f$. For each $y\in\{0,1\}^\kappa\backslash\{k_2\}$, the probability that $g(y)\in S_f$ is $\frac{|S_f|}{2^{tn}}$. Thus, the probability that there exists some $y\in\{0,1\}^\kappa\backslash\{k_2\}$ such that $g(y)\in S_f$ is approximated by $1-\left(1-\frac{|S_f|}{2^{tn}}\right)^{2^\kappa-1}\approx \frac{|S_f|(2^\kappa-1)}{2^{tn}} \leq 2^{2\kappa-tn}$. When $t=\lceil 2\kappa/n\rceil+1$, the probability of there existing a claw $(x,y)$ with $y\neq k_2$ such that $f(x)=g(y)$ is at most $2^{-n}$.

Analogously, we can show that the probability of there existing a claw $(x,y)$ with $x\neq k_1$ satisfying $f(x)=g(y)$ is at most $2^{-n}$. Consequently, the probability of there being another claw $(x,y)$ distinct from $(k_1,k_2)$ is at most $2^{-n+1}$. Thus concludes the proof.
\end{proof}

By Propositions \ref{prop1} and \ref{prop2}, the true key $(k_1,k_2)$ is essentially the unique claw satisfying $f(k_1)=g(k_2)$ when $t=\lceil 2\kappa/n\rceil+1$. The QCF-based MITM attack thus proceeds as follows.
\begin{description}
  \item[(1)] Define the functions $f$ and $g$, reducing the MITM attack to a claw-finding problem for $f$ and $g$. Implement the quantum Oracles $O_f$ and $O_g$ using quantum queries to $2kTE_{k_1,k_2}$ and the quantum circuits of $E$ and $E^{-1}$.
  \item[(2)] Execute the QCF algorithm using the quantum Oracles $O_f$ and $O_g$, and finally retrieve a claw $(x,y)$. The tuple $(x,y)$ corresponds to the true key $(k_1,k_2)$ with high probability.
  \item[(3)] Verify the correctness of $(x,y)$ using a fresh plaintext-ciphertext pair $(m',c')$. If $2kTE_{x,y} (m')=c'$, output $(x,y)$ as the final key.
\end{description}

The theoretical feasibility of this attack hinges on the attacker's access to quantum Oracles $O_f$ and $O_g$. Since $g(y)$ is defined by the block cipher $E_y(\cdot)$, $O_g$ can be implemented via $E$'s quantum circuit. In contrast, $f(x)$ depends on both $E_x^{-1}(\cdot)$ and $2kTE_{k_1,k_2}(\cdot)$, where $2kTE_{k_1,k_2}(\cdot)$ requires online quantum queries (as $k_1,k_2$ are unknown). This restricts the attack to the Q2 model, which permits quantum superposition queries to $2kTE_{k_1,k_2}(\cdot)$.

The attack complexity is dominated by the QCF algorithm. As detailed in Section \ref{sec23}, the QCF-based MITM attack achieves $O(2^{2\kappa/3})$ time and $O(2^{2\kappa/3})$ QRAM in Q2 model -- significantly outperforming classical MITM ($O(2^\kappa)$ \cite{merkle1981}) and Grover-based brute-force attacks ($O(2^\kappa)$), which demonstrates its efficiency gain in the Q2 model.

\subsection{Quantum Attack Based on Grover's Algorithm}\label{sec32}
In Section \ref{sec31}, the MITM attack on 2kTE is reduced to a claw-finding problem, which is then solved using QCF algorithm. In this section, we substitute Grover's algorithm for QCF to devise a novel quantum MITM attack, which necessitates higher QRAM resources. A time-QRAM tradeoff variant is also presented to optimize resource consumption.
\subsubsection{Quantum Attack with Unbounded QRAM}\label{sec321}
The Grover-based quantum MITM attack requires substantial QRAM to store classical preprocessing results. The core approach is as follows: First, the attacker precomputes and stores a table $L=\{(g(y),y)|y\in\{0,1\}^\kappa\}$; Then, the attacker searches for an $x\in\{0,1\}^\kappa$ such that $f(x)\in L^1$, where $L^1=\{g(y)|y\in\{0,1\}^\kappa\}$ denotes the set of first components of $L$.

The detailed process is described as follows.
\begin{description}
  \item[(1)] Classical preprocessing.
        \begin{description}
          \item[(1a)] Perform $t2^\kappa$ sequential encryptions $E_y(i)$ for all $y\in\{0,1\}^\kappa$ and $i\in[t]$, and store results in table $L=\{(g(y),y)|y\in\{0,1\}^\kappa\}$;
          \item[(1b)] Sort $L$ by the values of $g(y)$, retaining the notation $L$ for the sorted table.
        \end{description}
  \item[(2)] Key-searching process. It consists of two steps.
        \begin{description}
          \item[(2a)] Execute Grover's algorithm for $f(x)$ and $L^1$, which is denoted as $GroverSearch(f,L^1)$. It should find $x\in\{0,1\}^\kappa$ such that $f(x)\in L^1$. In other words, it would find $x$ such that $F(x)=1$, where $F:\{0,1\}^\kappa\rightarrow \{0,1\}$ is defined as follow.
               \begin{equation}
                F(x)=\left\{
                  \begin{array}{ll}
                    1, & \hbox{if } f(x)\in L^1; \\
                    0, & \hbox{others.}
                  \end{array}
                \right.
               \end{equation}
          \item[(2b)] If Grover's algorithm outputs a value $x'$ such that $f(x')\in L^1$, the attacker executes classical search algorithm to find the element $(g(y'),y')\in L$ such that $g(y')=f(x')$, and then outputs $(x',y')$ as a candidate key.
        \end{description}
  \item[(3)] Verify the correctness of $(x',y')$ using a fresh plaintext-ciphertext pair $(m',c')$. If $2kTE_{x',y'} (m')=c'$, output $(x',y')$ as the final key.
\end{description}

The theoretical feasibility of this attack hinges on the attacker's ability to access the quantum oracle $O_F$. Implementing $O_F$ requires satisfying two conditions: (1) The table $L$ (or $L^1$) must be stored in QRAM to enable quantum-level membership checks for $f(x)\in L^1$; (2) The computation of $f(x)$ must support quantum superposition, allowing the attacker to evaluate $O_F$ in superposition. Because $f(x)$ depends on both $E_x^{-1}(\cdot)$ and $2kTE_{k_1,k_2}(\cdot)$, and $2kTE_{k_1,k_2}(\cdot)$ can only be queried online due to unknown keys $k_1,k_2$, the attack is feasible exclusively in the Q2 model, which permits quantum superposition queries to $2kTE_{k_1,k_2}(\cdot)$.

Analogous to the reasoning in Section \ref{sec31}, when $t=\lceil 2\kappa/n\rceil+1$, the true key $(k_1,k_2)$ is essentially the unique claw satisfying $f(k_1)=g(k_2)$. Consequently, the candidate $(x',y')$ obtained in Step \textbf{(2)} matches the true key with high probability, ensuring the attack's high success rate.

Finally, we analyze the complexity of the aforementioned quantum MITM attack in detail.

The classical preprocessing procedure \textbf{(1)} consists of two steps: \textbf{(1a)} and \textbf{(1b)}. In Step \textbf{(1a)}, there are $t2^\kappa$ instances of block encryption; In Step \textbf{(1b)}, the attacker sorts the table $L$ of size $|L|=2^\kappa$, with a time complexity of $O(\kappa 2^\kappa)$. Hence, the time complexity of classical preprocessing procedure is
$O(t2^\kappa)+O(\kappa 2^\kappa)\approx O(\kappa 2^\kappa)$, where $t$ is typically negligible compared to $\kappa$.

In Step \textbf{(2)}, to execute the quantum algorithm $GroverSearch(f,L^1)$, the attacker must store the table $L^1$ in QRAM. Since $|L^1|=2^\kappa$, the algorithm requires QRAM of $O(2^\kappa)$. When there exists a unique claw $(k_1,k_2)$ satisfying $f(k_1)=g(k_2)$, there is exactly one $x$ in space $\{0,1\}^\kappa$ for which $F(x)=1$. Thus, the Grover iteration must be repeated $\lceil\frac{\pi}{4}\sqrt{2^\kappa}\rceil$ times. During each Grover iteration, the primary cost stems from the validating $f(x)\in L^1$. For a sorted table $L^1$, the validation complexity is $O(log|L^1|)=O(\kappa)$. Hence, the time complexity of Step \textbf{(2a)} is $O(\kappa 2^{\kappa/2})$. In Step \textbf{(2b)}, the attacker performs a classical search algorithm on the sorted table $L$, with a time complexity of $O(\kappa)$. The Step \textbf{(2)} can be completed in $O(\kappa 2^{\kappa/2})$ time.

In Step \textbf{(3)}, the attacker merely verifies the correctness of a candidate key, and its time complexity is negligible.

From the foregoing analysis, the classical preprocessing procedure concludes in $O(\kappa 2^\kappa)$ time. Notably, this procedure is independent of the true key $(k_1,k_2)$ and depends solely on the underlying block cipher, so this procedure can be completed at the first time of this attack. Once the table $L$ is fully constructed, it can be reused in subsequent attacks without requiring re-preprocessing. Consequently, preprocessing complexity is excluded from the overall quantum attack analysis.
The complexity of the aforementioned quantum MITM attack is dominated by Step \textbf{(2)}, achieving a time complexity of $O(\kappa 2^{\kappa/2})$ with $O(2^\kappa)$ QRAM.

Classical MITM \cite{merkle1981} and Grover-based brute-force attacks exhibit $O(2^\kappa)$ time complexity, while the QCF-based MITM attack in Section \ref{sec31} achieves $O(2^{2\kappa/3})$. The Grover-based MITM attack reduces time complexity to $O(\kappa 2^{\kappa/2})$ but requires $O(2^\kappa)$ QRAM -- substantially exceeding the time cost. To address this issue, we propose a time-QRAM tradeoff scheme.

\subsubsection{A Time-QRAM Tradeoff Variant}
Based on the quantum attack in Section \ref{sec321}, the required QRAM memory size is $O(2^\kappa)$, and the quantum time complexity is $O(\kappa 2^{\kappa/2})$. Given the current lack of effective QRAM implementations, table $L$ must be stored in qubits for use in the quantum algorithm. This incurs significant costs due to the extremely large size of $|L|$. To optimize QRAM usage, we propose a quantum attack with a tradeoff between time complexity and QRAM. The core concept is to partition table $L$ into $r$ sub-tables of equal size, which are sequentially utilized in multiple executions of Grover's algorithm. The detailed attack is as follows. To simplify the description, we only elaborate on the differences from the previous quantum attack.
\begin{description}
  \item[(1)] Classical preprocessing procedure. It comprises three steps, where steps \textbf{(1a)} and \textbf{(1b)} remain identical to the attack in Section \ref{sec321}, and step \textbf{(1c)} is as follows.
        \begin{description}
          \item[(1c)] Partition table $L$ into $r$ sub-tables $L_i(i\in[r])$ of equal size, such that $|L_i|=2^\kappa/r$. Define $L_i^1$ as the table comprising the first value of each element in $L_i$, i.e., $L_i^1=\{g(y)|(g(y),y)\in L_i\}$.
        \end{description}
  \item[(2)] Key-searching process. This process mirrors the previous approach but executes the following step \textbf{(2a')} in place of \textbf{(2a)}.
        \begin{description}
          \item[(2a')] For $i=1,2,\ldots,r$, sequentially perform $GroverSearch(f,L_i^1)$ until an $x\in\{0,1\}^\kappa$ is found such that $f(x)$ resides in one of the sub-tables $L_i^1$ for $i\in[r]$.
        \end{description}
  \item[(3)] Identical to Step \textbf{(3)} in Section \ref{sec321}.
\end{description}

For the aforementioned tradeoff scheme, during each execution of Grover's algorithm $GroverSearch(f,L_i^1)$, only the sub-table $L_i^1$ is stored in QRAM, with $|L_i^1|=2^\kappa/r$. Thus, this attack requires QRAM of size $M=O(2^\kappa/r)$.

When there exists a unique claw $(k_1,k_2)$ satisfying $f(k_1)=g(k_2)$, there is exactly one $x$ in $\{0,1\}^\kappa$ such that $f(x)\in L^1$, and $f(x)$ resides in exactly one sub-table $L_j^1$ for $j\in[r]$. In the proposed attack, each execution of $GroverSearch(f,L_i^1)$ searches for a target within $\{0,1\}^\kappa$,  with at most one target $x$ satisfying $f(x)\in L_i^1$. Thus, the Grover iteration repeats $\lceil\frac{\pi}{4}2^{\kappa/2}\rceil$ times. Additionally, Step \textbf{(2a')} invokes $GroverSearch(f,L_i^1)$ an unspecified number of times:
\begin{itemize}
  \item \textbf{Best case}: The first call finds $x$ such that $f(x)\in L_1^1$;
  \item \textbf{Worst case}: The $r$-th call finds $x$ such that $f(x)\in L_r^1$;
  \item \textbf{Average case}: $r/2$ calls are required.
\end{itemize}
Consequently, the average time complexity is $T=O(r2^{\kappa/2})$.

In summary, the tradeoff scheme requires QRAM of size $M=O(2^\kappa/r)$, with a time complexity of $T=O(r2^{\kappa/2})$. Setting $r=2^{\kappa/4}$ yields both QRAM and time complexity of $O(2^{3\kappa/4})$, outperforming the classical MITM attack \cite{merkle1981} and Grover-based brute-force attack -- both of which demand $O(2^\kappa)$ time. Compared to the QCF-based attack in Section \ref{sec31}, the tradeoff scheme exhibits higher time complexity and larger QRAM requirements. Nevertheless, the security analysis of 3XCE in Section \ref{sec4} indicates that the Grover-based MITM attack may surpass the QCF-based approach.

\subsection{Applications to Two-Key 3DES and Generalized Two-Key Triple Encryption}\label{sec33}
The quantum MITM attacks described in Sections \ref{sec31} and \ref{sec32} are applicable to attacking two-key 3DES and generalized 2kTE. We present the cryptanalysis results in detail.
\subsubsection{Two-Key 3DES}
3DES consists of two DES encryption calls and one DES decryption call. Two-key 3DES, a special case of 3DES with two independent keys, can be formally described as follows:
\begin{equation}
2k3DES_{k_1,k_2}(m)=DES_{k_1}(DES_{k_2}^{-1}(DES_{k_1}(m))),
\end{equation}
where $k_1,k_2\in\{0,1\}^{56}$ and $m\in\{0,1\}^{64}$.

Define two functions $f,g:\{0,1\}^\kappa\rightarrow \{0,1\}^{3n}$ as follows:
\begin{eqnarray}
f(x)=&&DES_x^{-1}\left(2k3DES_{k_1,k_2}\left(DES_x^{-1}(1)\right)\right)\parallel\cdots\nonumber\\
 &&\cdots DES_x^{-1}\left(2k3DES_{k_1,k_2}\left(DES_x^{-1}(2)\right)\right)\parallel\cdots\nonumber\\
 &&\cdots DES_x^{-1}\left(2k3DES_{k_1,k_2}\left(DES_x^{-1}(3)\right)\right),\\
g(y)=&& DES_y^{-1}(1)\parallel DES_y^{-1}(2)\parallel DES_y^{-1}(3).
\end{eqnarray}

Drawing from Merkle-Hellman's attack \cite{merkle1981}, the classical MITM attack on two-key 3DES exhibits a time complexity of $O(2^{56})$. When Grover's algorithm is employed for brute-force key search, the time complexity also amounts to $O(2^{56})$. Applying the quantum MITM attacks detailed in Sections \ref{sec31} and \ref{sec32} yields significant improvements in time complexity.

For the QCF-based MITM attack outlined in Section \ref{sec31}, two-key 3DES can be broken within the Q2 model with a time complexity of $O(2^{37})$ using $O(2^{37})$ QRAM. As for the Grover-based MITM attack in Section \ref{sec32}, breaking two-key 3DES in the Q2 model requires a time complexity of $O(2^{28})$ with $O(2^{56})$ QRAM. By optimizing QRAM through the tradeoff scheme (setting $r=2^{14}$), both the time complexity and QRAM size are reduced to $O(2^{42})$.

\subsubsection{Generalized Two-Key Triple Encryption}
For the generalized two-key triple encryption (G2kTE) scheme, all three encryption calls are independent, and the quantum MITM attacks in Sections \ref{sec31} and \ref{sec32} remain applicable.

The G2kTE scheme is defined as:
\begin{equation}
G2kTE_{k_1,k_2}(m)=E_{k_1}^3 \left(E_{k_2}^2 \left(E_{k_1}^1 (m)\right)\right),
\end{equation}
where $E^1,E^2,E^3\in BC(\kappa,n)$ denote three independent block ciphers, and $k_1,k_2\in\{0,1\}^\kappa$ are two independent secret keys. The first key $k_1$ is used in both $E^1$ and $E^3$, while the second key $k_2$ is used in $E^2$.

Define two functions $f,g:\{0,1\}^\kappa\rightarrow \{0,1\}^{tn}$ as follows:
\begin{eqnarray}
f(x)=&&D_x^{3}\left(G2kTE_{k_1,k_2}\left(D_x^{1}(1)\right)\right)\parallel D_x^{3}\left(G2kTE_{k_1,k_2}\left(D_x^{1}(2)\right)\right)\parallel\cdots\nonumber\\
 &&\cdots D_x^{3}\left(G2kTE_{k_1,k_2}\left(D_x^{1}(t)\right)\right),\\
g(y)=&& E_y^{2}(1)\parallel E_y^{2}(2)\parallel\cdots\parallel E_y^{2}(t),
\end{eqnarray}
where $t=\lceil 2\kappa/n\rceil+1$, $k_1,k_2$ are the true keys for the G2kTE scheme, and $D^i(i=1,2,3)$ are the decryption algorithms corresponding to the encryptions $E^i(i=1,2,3)$.

For the G2kTE scheme, the classical MITM attack \cite{merkle1981} exhibits a time complexity of $O(2^\kappa)$, where $\kappa$ denotes the length of $k_1$. When Grover's algorithm is employed for brute-force key search of $k_1$ and $k_2$, the time complexity is $O(\sqrt{2^{2\kappa}})=O(2^\kappa)$. In comparison to the classical MITM attack and Grover-based brute-force attack, the quantum MITM attacks detailed in Sections \ref{sec31} and \ref{sec32} yield significant improvements in time complexity.

For the QCF-based MITM attack in Section \ref{sec31}, the G2kTE scheme can be broken in the Q2 model with a time complexity of $O(2^{2\kappa/3})$ and QRAM of $O(2^{2\kappa/3})$. As for the Grover-based MITM attack in Section \ref{sec32}, breaking the G2kTE scheme in the Q2 model requires a time complexity of $O(2^{\kappa/2})$ and QRAM of $O(2^\kappa)$. By adopting the tradeoff scheme and setting $r=2^{\kappa/4}$, both the time complexity and QRAM size are reduced to $O(2^{3\kappa/4})$.

\section{3XOR-Cascade Encryption}\label{sec4}
The 3XOR-cascade encryption (3XCE) scheme comprises three whitening keys and two block encryption calls (see Fig. \ref{fig2}). It can be formalized as:
\begin{equation}
3XCE_{k,k_1,k_2}(m)=E_k^2(E_k^1(m\oplus k_1)\oplus k_2)\oplus k_1,
\end{equation}
where $E^1,E^2\in BC(\kappa,n)$ denote two independent block ciphers, $k\in\{0,1\}^\kappa$ and $k_1,k_2\in\{0,1\}^n$ are three independent secret keys. The key $k$ serves as the encryption key for both $E^1$ and $E^2$, while the other two keys $k_1,k_2$ are whitening keys.

\begin{figure}[h]
\centering
  \includegraphics[scale=0.23]{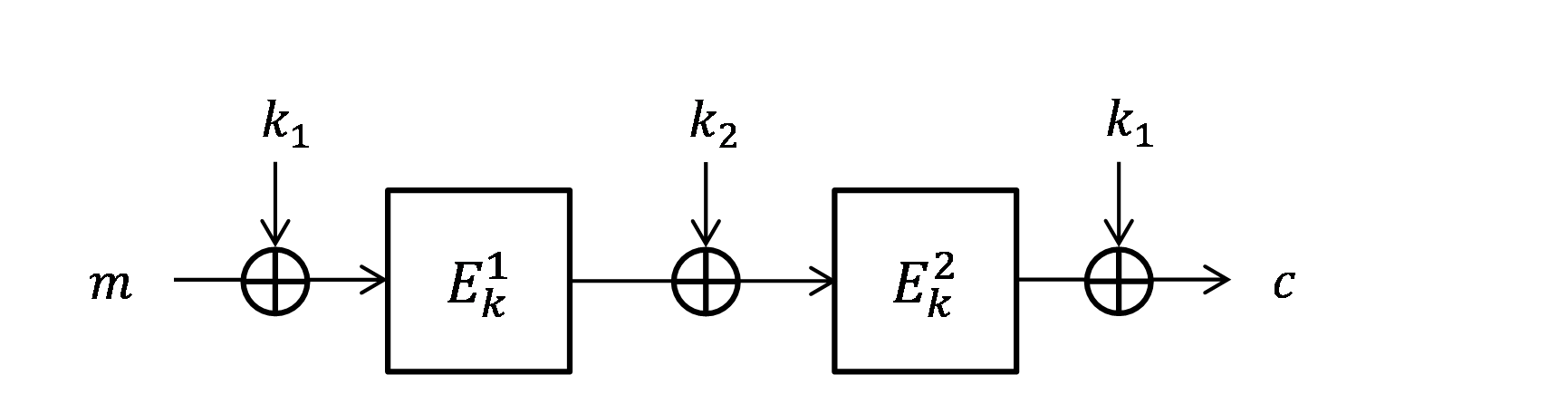}
\caption{3XOR-cascade encryption.}
\label{fig2}       
\end{figure}

The classical Merkle-Hellman's attack can be applied to analyze the 3XCE scheme, as 3XCE can be rewritten as Eq. (\ref{eqn3xce}). Consequently, the classical MITM attack \cite{merkle1981} can break 3XCE with a time complexity of $O(2^{\kappa+n})$, where $\kappa$ and $n$ denote the length of the keys $k$ and $k_1$, respectively. When Grover's algorithm is employed for brute-force key search, the time complexity is $O(2^{(\kappa+2n)/2})$, with $\kappa+2n$ being the total length of all secret keys.

We next propose quantum MITM attacks on $3XCE_{k,k_1,k_2}$ in both the Q1 and Q2 models. These attacks significantly reduce the time complexity compared to the classical MITM attack and Grover-based brute-force attack.

\subsection{Security Analysis in the Q2 Model}\label{sec41}
The 3XCE scheme can be regarded as a special case of the G2kTE scheme, and its analysis in the Q2 model can be conducted similarly to that of G2kTE. The detailed analysis is as follows.

According to the construction of the G2kTE scheme, we define:
\begin{eqnarray}
\tilde{E}_{k,k_1}^1 (x)&&=E_k^1 (x\oplus k_1),\\
\tilde{E}_{k_2}^2 (x)&&=\oplus_{k_2} (x)=x\oplus k_2,\\
\tilde{E}_{k,k_1}^3 (x)&&=E_k^2 (x)\oplus k_1.
\end{eqnarray}
Then the 3XCE scheme can be rewritten as:
\begin{equation}\label{eqn3xce}
3XCE_{k,k_1,k_2}(m)=\tilde{E}_{k,k_1}^3 \left(\tilde{E}_{k_2}^2 \left(\tilde{E}_{k,k_1}^1 (m)\right)\right).
\end{equation}
Thus, we can design quantum MITM attacks against 3XCE in the Q2 model by following a similar approach to that described in Section \ref{sec3}.

Although the 3XCE scheme can be characterized as a composition of three encryptions $\tilde{E}_{k,k_1}^1,\tilde{E}_{k_2}^2$ and $\tilde{E}_{k,k_1}^3$, the keys employed by these three components differ in size. Consequently, the attack complexity for 3XCE exhibits slight variations compared to that of G2kTE. We provide a detailed analysis of 3XCE within the Q2 model and compare it with other attack strategies.

For the QCF-based MITM attack: When $\kappa<n$, the time complexity and QRAM requirement are both $O((2^{\kappa+n}\cdot 2^n )^{1/3})=O(2^{(\kappa+2n)/3} )$; When $\kappa>n$, both the time complexity and QRAM size amount to $O(2^{(\kappa+n)/2})$.

For the Grover-based MITM attack: The attack achieves a time complexity of $O(2^{(\kappa+n)/2})$ with $O(2^n)$ QRAM; Using the tradeoff scheme with QRAM size $O(2^n/r)$, the time complexity reduces to $O(2^{(\kappa+n)/2}r)$. Setting $r=2^{(n-\kappa)/4}$ yields both time complexity and QRAM size of $O(2^{(3n+\kappa)/4})$.

The above analysis is summarized in Table.\ref{table1}, facilitating a comparison between the two quantum attack approaches. If $\kappa <n$, since $(\kappa +n)/2<(\kappa +2n)/3$, the Grover-based MITM attack outperforms the QCF-based variant; If $\kappa >n$, as $(3n+\kappa)/4<(\kappa +n)/2$, the tradeoff scheme proves superior to the QCF-based attack. Thus, the QCF-based MITM attack offers no advantage over the Grover-based MITM attack when targeting 3XCE in the Q2 model.
\begin{table}[h]
  \centering
  \caption{The attack complexity of the 3XCE scheme under the Q2 model.}\label{table1}
  \begin{tabular}{|l|c|c|c|}
    \hline
    Quantum MITM Attacks & Cases & Time Complexity &  QRAM \\
    \hline
    QCF-Based Attack & $\begin{array}{c} \kappa <n \\ \kappa >n \end{array}$ &  $\begin{array}{c} O(2^{(\kappa+2n)/3}) \\ O(2^{(\kappa+n)/2}) \end{array}$ &  $\begin{array}{c} O(2^{(\kappa+2n)/3}) \\ O(2^{(\kappa+n)/2}) \end{array}$ \\
    \hline
    Grover-Based Attack &  $\begin{array}{c} \hbox{Unbounded QRAM} \\ \hbox{Tradeoff Variant} \end{array}$ & $\begin{array}{c} O(2^{(\kappa+n)/2}) \\ O(2^{(\kappa+3n)/4}) \end{array}$ & $\begin{array}{c} O(2^n) \\ O(2^{(\kappa+3n)/4}) \end{array}$ \\
    \hline
  \end{tabular}
\end{table}

\subsection{Quantum Meet-in-the-Middle Attack in the Q1 Model}\label{sec42}
Notice that $\tilde{E}_{k_2}^2$ in Eq. (\ref{eqn3xce}) involves only an XOR operation ($\tilde{E}_{k_2}^2(x)=x\oplus k_2$). Thus, any input-output pair $(a,b)$ for $\tilde{E}_{k_2}^2$ satisfies $a\oplus b\equiv k_2$, where the key $k_2$ is an unknown fixed value for attackers. Leveraging this property, we propose a quantum MITM attack that requires no QRAM or classical preprocessing, suitable for targeting the 3XCE scheme in the Q1 model.

The core idea is as follows. Given $t=\lceil(\kappa+2n)/n\rceil+1$ plaintext-ciphertext pairs $\{(m_i,c_i),i\in[t]\}$, the attacker guesses keys $k$ and $k_1$ (denoted as $x$ and $y$ for guessing values), computes $a_i=E_x^1 (m_i\oplus y)$ and $b_i=D_x^2 (c_i\oplus y )$ for $i\in[t]$, and checks whether $a_i\oplus b_i=a_j\oplus b_j$ for all $1\leq i<j\leq t$. If any $i,j$ satisfies $a_i\oplus b_i\neq a_j\oplus b_j$, the guessed values $(x,y)$ are incorrect; otherwise, they are likely the true keys $k,k_1$ with high probability. Finally, two additional plaintext-ciphertext pairs $\{(m_i',c_i')|i\in\{1,2\}\}$ are used to verify the guess.

The detailed quantum MITM attack on 3XCE is as follows.
\begin{description}
  \item[(1)] Collect $t$ plaintext-ciphertext pairs $\{(m_i,c_i),i\in[t]\}$.
  \item[(2)] Execute Grover's algorithm to find $(k',k_1')$ such that $F(k',k_1')=1$, where the function $F:\{0,1\}^{\kappa+n}\rightarrow \{0,1\}$ is defined as:
               \begin{equation}
                F(x,y)=\left\{
                  \begin{array}{ll}
                    1, & \delta_i=\delta_j \texttt{ for all } i,j\in[t], \\
                    0, & \delta_i\neq\delta_j \texttt{ for some } i,j\in[t],
                  \end{array}
                \right.
               \end{equation}
        where $\delta_i = E_x^1 (m_i\oplus y)\oplus D_x^2 (c_i\oplus y)$ for $i\in[t]$. For each candidate key $(k',k_1')$, $\delta_i$ is computed from $(m_i,c_i)$ to check $F(x,y)=1$.
  \item[(3)] Verify candidate keys $(k',k_1')$ using two new pairs $\{(m_i',c_i'),i\in\{1,2\}\}$. Compute $\delta_i'=E_{k'}^1 (m_i'\oplus k_1')\oplus D_{k'}^2 (c_i'\oplus k_1')$ for $i\in\{1,2\}$. If $\delta_1'=\delta_2'$, set $k_2'=\delta_1'$ and return $(k',k_1',k_2')$; Otherwise, retry from Step \textbf{(1)}.
\end{description}

Compared with the attack in Section \ref{sec41}, this quantum MITM attack offers three advantages: 1) It requires no quantum queries to $3XCE_{k,k_1,k_2}$, making it applicable in the Q1 model; 2) It omits classical preprocessing, ensuring simplicity; 3) It eliminates QRAM requirements for storing preprocessed values.

The main complexity is attributed to Step \textbf{(2)}. Leveraging Grover's algorithm, the time complexity amounts to $O(2^{(\kappa+n)/2})$, which achieves a quadratic speedup over its classical counterpart ($O(2^{\kappa+n})$). When compared to the Grover-based brute-force attack ($O(2^{(\kappa+2n)/2})$), this attack in the Q1 model also demonstrates higher efficiency.

\section{Quantum Sieve-in-the-Middle Attack}\label{sec5}
The sieve-in-the-middle (SITM) attack \cite{canteaut2013} represents a variant of the MITM attack, applicable to a broader range of cryptographic constructions. Unlike the MITM attack, which relies on collisions between two intermediate states, the SITM attack sieves candidate keys using other easy-to-verify conditions, requiring the attacker to check whether intermediate states satisfy these conditions.

\subsection{General Framework}\label{sec51}
Inspired by the quantum MITM attack described in Section \ref{sec42}, we propose a general framework for the quantum SITM attack on the construction:
\begin{equation}
ELE_{K_1,K_2}=E_{K_1}^2\circ L_{K_2}\circ E_{K_1}^1,
\end{equation}
which generalizes the 3XCE construction. Here, $E^1$ and $E^2$ are two block ciphers parameterized by the key $K_1\in\mathcal{K}$, and $L$ is a family of permutations dependent on the key $K_2\in\{0,1\}^n$. Notably, $L$ does not need to be a block cipher -- for example, $L_{K_2}(x)=x\oplus K_2$.

\noindent\textbf{Assumption.} Suppose the family of permutations $\{L_k,k\in\{0,1\}^n\}$ satisfies the following condition: for any unknown key $k$ and a sufficiently large constant $t$, there exists a distinguisher $\mathcal{A}$ that can determine in $O(T)$ time whether the input $S=\{(a_i,b_i)|i\in[t]\}$ is derived from the set $\{(x_i,y_i)|y_i=L_{k}(x_i),\forall i\in[t]\}$.

To attack $ELE_{K_1,K_2}$, we extend the quantum MITM attack in Section \ref{sec42} and propose the following quantum SITM attack.
\begin{description}
  \item[(1)] Design a distinguisher $\mathcal{A}(S)$ according to the permutations $\{L_k,k\in\{0,1\}^n\}$. If the input $S$ is derived from $\{(x_i,y_i)|y_i=L_{k}(x_i),\forall i\in[t]\}$, $\mathcal{A}$ returns 1; Otherwise $\mathcal{A}$ returns 0.
  \item[(2)] Collect $t$ plaintext-ciphertext pairs $(m_i,c_i),i\in[t]$, then define the function $F:\mathcal{K}\rightarrow \{0,1\}$ using the distinguisher $\mathcal{A}$:
               \begin{equation}\label{equ518}
                F(x)=\left\{
                  \begin{array}{ll}
                    1, & \hbox{if }\mathcal{A}\left(S_{x}\right)=1, \\
                    0, & \hbox{if }\mathcal{A}\left(S_{x}\right)=0,
                  \end{array}
                \right.
               \end{equation}
            where $S_{x}=\{\left(E_x^1(m_i),D_x^2(c_i)\right),i\in[t]\}$.
  \item[(3)] Execute a quantum attack analogous to that in Section \ref{sec42}, but use the redefined function $F$ in Step \textbf{(2)}.
\end{description}

By the definition of $F$, the function $F$ can be implemented via $\mathcal{A}$ with a time complexity of $O(T)$. Thus, the quantum SITM attack achieves a time complexity of $O(T\cdot \sqrt{|\mathcal{K}|})$.

To demonstrate the general framework, we consider a quantum SITM attack on a class of constructions $\widetilde{3XCE}$ (see Fig. \ref{fig3}), where the $\oplus_{k_2}$ operation in the 3XCE scheme is replaced by a more general permutation $L_{k_2}$.
\begin{figure}[h]
\centering
  \includegraphics[scale=0.23]{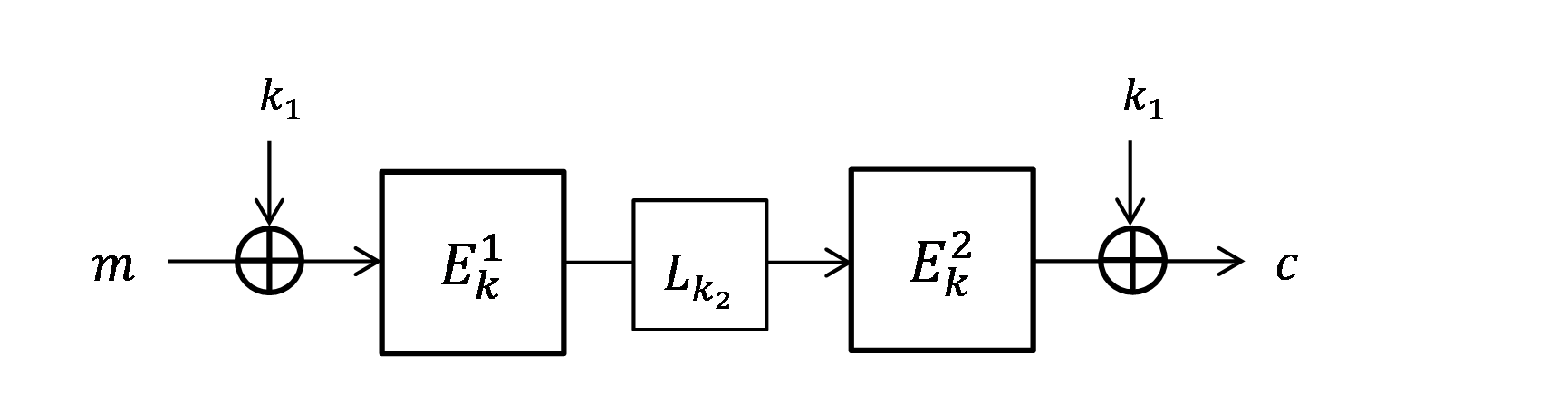}
\caption{The construction $\widetilde{3XCE}$. The middle layer $L_{k_2}$ is a permutation.}
\label{fig3}       
\end{figure}

The construction is formally described as follows:
\begin{equation}
\widetilde{3XCE}_{k,k_1,k_2}(m)=E_k^2\left(L_{k_2}\left(E_k^1(m\oplus k_1)\right)\right)\oplus k_1.
\end{equation}
We propose a quantum SITM attack on the $\widetilde{3XCE}$ construction, where the middle layer can take one of three different forms.

\subsection{Applications to Two Concrete Constructions}\label{sec52}
In this section, we introduce two concrete constructions to illustrate the applications of the quantum SITM attack. For each construction, the middle layer $L$ is sufficiently simple, allowing the distinguisher $\mathcal{A}$ to be implemented in $O(1)$ time.

\subsubsection{Application to 3XCE}
The 3XCE construction is a special case of $\widetilde{3XCE}$ (with $L_{k_2}(x)=x\oplus k_2$) and can be targeted by the quantum SITM attack.

Since $L_{k_2}(x')\oplus L_{k_2}(x)=x'\oplus x$ for any $k_2$, the output difference of $L_{k_2}$ must equal the input difference, independent of the specific value of $k_2$.

The attacker guesses the keys $k$ and $k_1$ (denoted as $x$ and $y$ for guessed values) and computes a set $S_{x,y}=\{(a_i,b_i ),i\in[t]\}$, where $a_i=E_x^1 (m_i\oplus y)$ and $b_i=D_x^2 (c_i\oplus y)$.
Given the input $S_{x,y}$, the distinguisher $\mathcal{A}$ checks whether $b_i\oplus b_j=a_i\oplus a_j$ for all $i,j\in[t]$, and determines if $S_{x,y}$ is derived from $\{(x_i,y_i)|y_i=L_{k_2} (x_i),\forall i\in[t]\}$. Since $t=\lceil(\kappa+2n)/n\rceil+1$ is typically a small constant, the time complexity of $\mathcal{A}$ is $O(1)$, making the total complexity of the quantum SITM attack on $3XCE_{k,k_1,k_2}$ equal to $O(2^{(\kappa+n)/2})$.

\subsubsection{Application to Two-Round Key-Alternating Reflection Cipher}
The quantum SITM attack can be applied to the key-alternating reflection cipher (KARC) construction proposed by Beyne and Chen \cite{beyne2022}. The KARC construction employs three independent keys $k,k_0$ and $k_1$, and is defined as follows (see Fig. \ref{fig4}):
\begin{equation}
KARC_{k,k_0,k_1}(m)=D_{k\oplus\alpha} \left(\mathcal{R}\left(E_k (m\oplus k_1 )\oplus k_2\right) \oplus \sigma(k_2)\right)\oplus \sigma(k_1),
\end{equation}
where $(E,D)$ denotes a block cipher, $\alpha$ is a constant, $\sigma$ is an invertible linear function, and $\mathcal{R}$ is a linear involution.
\begin{figure}[h]
\centering
  \includegraphics[scale=0.23]{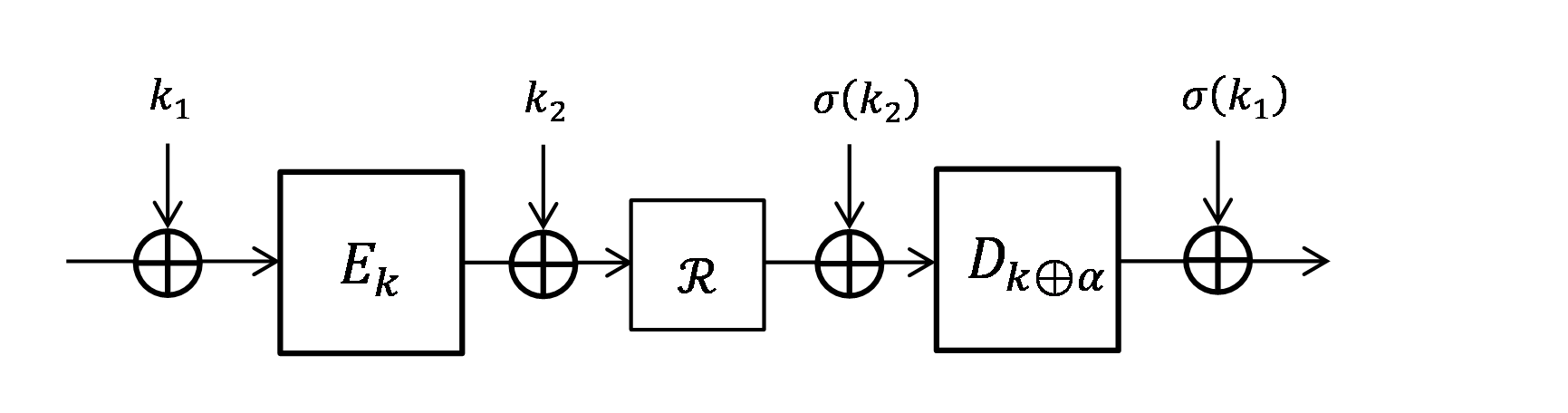}
\caption{Two-round key-alternating reflection cipher.}
\label{fig4}       
\end{figure}

Define $L_{k_2}(x)=\mathcal{R}(x\oplus k_2 )\oplus \sigma(k_2)$. For a set $S=\{(a_i,b_i)|i\in[t]\}$, if $b_i=L_{k_2}(a_i)$, then
\begin{equation}
b_i= \mathcal{R}(a_i\oplus k_2)\oplus \sigma(k_2)=\mathcal{R}(a_i)\oplus \mathcal{R}(k_2)\oplus \sigma(k_2).
\end{equation}
Thus, $b_i\oplus \mathcal{R}(a_i)= \mathcal{R}(k_2 )\oplus \sigma(k_2 )$ for all $i\in[t]$, indicating that $b_i\oplus \mathcal{R}(a_i)$ is a fixed value dependent on the secret key $k_2$.

The attacker guesses values for keys $k$ and $k_1$ (denoted as $x$ and $y$ for guessed values) and computes a set $S_{x,y}=\{(a_i,b_i ),i\in[t]\}$, where
$a_i=E_x (m_i\oplus y)$ and $b_i=E_{x\oplus\alpha} (c_i\oplus \sigma(y))$.
Given the input $S_{x,y}$, the distinguisher $\mathcal{A}$ (without knowledge of $k_2$) checks whether $b_i\oplus\mathcal{R}(a_i)=b_j\oplus\mathcal{R}(a_j )$ for all $i,j\in[t]$ to determine if $S_{x,y}$ is derived from $\{(x_i,y_i)|y_i=L_{k_2}(x_i),\forall i\in[t]\}$. Since $t$ is typically a small constant, the time complexity of $\mathcal{A}$ is $O(1)$, and the complexity of the quantum SITM attack is $O(2^{(\kappa+n)/2})$.

\begin{remark}
In the two special cases above, only a few pairs $(a_i,b_i)$ suffice for $\mathcal{A}$ to make a decision. For a complex permutation $L_{k_2}$, however, the distinguisher $\mathcal{A}$ may require a large amount of input data $(a_i,b_i)$, check whether these data satisfy specific cryptographic properties (e.g., differential or linear characteristics) with a probability $2^{-p}$ (where $p$ is large), and finally determine if they are derived from the set $\{(x_i,y_i)|y_i=L_{k_2}(x_i),\forall i\in[t]\}$.
\end{remark}

\subsection{Quantum SITM Attack Combined with Mirror Slide Attack}\label{sec53}
In Section \ref{sec52}, the middle layers in the concrete constructions are relatively simple. We now consider a more general case where $L$ in Fig. \ref{fig3} is an arbitrary involution (i.e., $L\circ L=Id$, where $Id$ denotes the identity transformation). In this scenario, the mirror slide attack \cite{dunkelman2015} can be integrated with the quantum SITM attack, where the distinguisher $\mathcal{A}$ leverages mirror slide pairs to make decisions.

Assume that for two plaintext-ciphertext pairs $(m,c),(m^*,c^*)$, the following holds:
\begin{equation}\label{equ521}
E_k^1(m\oplus k_1)=D_k^2(c^*\oplus k_1).
\end{equation}
Since $L$ is an involution, it follows that
\begin{equation}\label{equ522}
L_{k_2}\left(E_k^1(m\oplus k_1)\right)=L_{k_2}^{-1}\left(D_k^2(c^*\oplus k_1)\right).
\end{equation}
By the construction in Fig. \ref{fig3}, this implies that
\begin{equation}\label{equ523}
D_k^2(c\oplus k_1)=E_k^1(m^*\oplus k_1).
\end{equation}
If Eq. (\ref{equ521}) holds (and thus Eq. (\ref{equ523}) also holds), the pair $(m,m^*)$ is referred to as a \textbf{mirror slid pair}.

We propose two approaches to exploit mirror slide pairs in the following quantum SITM attack, one suitable for the Q1 model and the other for the Q2 model.

\subsubsection{Attack in the Q1 Model}
The attacker collects $t=2^{(n+1)/2}$ plaintext-ciphertext pairs $(m_i,c_i)$ for $i\in[t]$. With high probability, a mirror slid pair is expected to exist, satisfying Eqs. (\ref{equ521}) and (\ref{equ523}).

The attacker guesses keys $k$ and $k_1$ (denoted as $x$ and $y$ for guessed values) and computes
\begin{equation}
S_{x,y}=\{(a_i,b_i), a_i=E_x^1(m_i\oplus y), b_i=D_x^2(c_i\oplus y),i\in[t]\}.
\end{equation}
According to Eqs. (\ref{equ521}) and (\ref{equ523}), a distinguisher $\mathcal{A}(S_{x,y})$ is constructed to check if there exist indices $i,j\in[t](i\neq j)$ such that $a_i=b_j$ and $a_j=b_i$. If such indices exist, $\mathcal{A}$ returns 1; Otherwise it returns 0.
The quantum SITM attack described in Section \ref{sec51} can then be applied. However, since $\mathcal{A}$ requires $O(t^2)$ steps for the check, the attack's time complexity becomes $O(t^2\cdot 2^{(\kappa+n)/2})=O(2^{(\kappa+3n)/2})$, which is less efficient than the Grover-based brute-force attack.

This attack is a known-plaintext attack, and the attacker must collect $t=2^{(n+1)/2}$ plaintext-ciphertext pairs to ensure the high-probability existence of a mirror slide pair. This imposes significant overhead on $\mathcal{A}$, causing the attack to lose its advantage over the Grover-based brute-force approach. To address this, we will propose a quantum chosen-ciphertext attack in the Q2 model, where the new attack specifies partial information of the mirror slide pair. Once the attacker guesses the correct keys, a complete mirror slide pair can be computed.

\subsubsection{Attack in the Q2 Model}
In the Q2 model, the attacker is allowed to query $Enc_{k_1,k_2}$ and $Enc_{k_1,k_2}^{-1}$ in quantum superposition. Instead of collecting $2^{(n+1)/2}$ plaintext-ciphertext pairs, we specify that the values on both sides of Eq. (\ref{equ521}) equal $i\in[t]$, where $t=\lceil (\kappa+2n)/n\rceil +1$ ensures key uniqueness with high probability.

Assume $E_k^1(m\oplus k_1)=D_k^2(c^*\oplus k_1)=i$ for some $i\in\{0,1\}^n$ in Eq. (\ref{equ521}). Then
\begin{eqnarray}
&&m=D_k^1(i)\oplus k_1,\\
&&c^*=E_k^2(i)\oplus k_1.
\end{eqnarray}
The attacker queries $m$ to obtain ciphertext $c=Enc_{k_1,k_2}(m)$ and queries $c^*$ to obtain plaintext $m^*=Enc_{k_1,k_2}^{-1}(c^*)$.
By Eq. (\ref{equ523}), $c$ and $m^*$ must satisfy $D_k^2(c\oplus k_1)=E_k^1(m^*\oplus k_1)$. Based on this property, we can propose a quantum SITM attack in the Q2 model. This attack is essentially consistent with the quantum SITM framework. To avoid redundancy, only the key points of this attack are described below.

The attacker guesses keys $k,k_1$ (denoted as $x,y$ for guessed values) and computes the set
    \begin{equation}
        S_{x,y}=\{(a_i,b_i), i\in[t]\},
    \end{equation}
    where $a_i=D_x^2(Enc_{k_1,k_2}(D_x^1(i)\oplus y)\oplus y)$ and $b_i=E_x^1(Enc_{k_1,k_2}^{-1}(E_x^2(i)\oplus y)\oplus y)$.
Because the attacker does not know the true keys $k,k_1$, the computation of $a_i,b_i$ is implemented by online queries to $Enc_{k_1,k_2}$ and $Enc_{k_1,k_2}^{-1}$.

Define a distinguisher $\mathcal{A}(S_{x,y})$, which checks if $a_i=b_i$ for all $i\in[t]$. It returns 1 if true; Otherwise returns 0. The time complexity of $\mathcal{A}$ is $O(1)$.

Define $F(x,y)$ as in Eq. (\ref{equ518}). The attacker employs Grover's algorithm to search for keys satisfying $F(x,y)=1$. The quantum SITM attack incurs a time complexity of $O(2^{(\kappa+n)/2})$, outperforming the Grover-based brute-force attack.

\begin{remark}
The requirement that $L$ is an involution can be relaxed to $L\circ P \circ L=Id$ for a key-independent operation $P$. In this case, this attack only requires simple modifications. Firstly, Eq. (\ref{equ522}) becomes $P\circ L_{k_2}\left(E_k^1(m\oplus k_1)\right)=L_{k_2}^{-1}\left(D_k^2(c^*\oplus k_1)\right)$, and
Eq. (\ref{equ523}) becomes $P\circ D_k^2(c\oplus k_1)=E_k^1(m^*\oplus k_1)$. Thus, the distinguisher $\mathcal{A}$ should check $P(a_i)=b_i$ instead of $a_i=b_i$.
\end{remark}

\section{Discussions and Conclusions}
This paper systematically evaluates the quantum security of two KLE constructions -- 2kTE and 3XCE -- through quantum MITM attacks. By leveraging quantum algorithms such as Grover's search and QCF, we demonstrate significant vulnerabilities in both constructions under quantum models.

\subsection{Quantum Security Analysis of 2kTE}
For 2kTE, we propose two quantum MITM attacks in the Q2 model, where attackers can leverage quantum superposition queries.
\begin{description}
  \item[\textbf{QCF-Based Attack}:] With a time complexity of $O(2^{2\kappa/3})$ and QRAM requirement of $O(2^{2\kappa/3})$, this attack outperforms classical MITM ($O(2^\kappa)$) and quantum brute-force attacks ($O(2^\kappa)$) by exploiting claw-finding problems.
  \item[\textbf{Grover-Based Attack}:] Achieving a time complexity of $O(2^{\kappa/2})$ (quadratic speedup over classical methods), this approach requires $O(2^\kappa)$ QRAM. Notably, when QRAM is sufficiently large, 2kTE offers no security enhancement over the underlying block cipher in Q2 model, as its quantum attack complexity matches that of direct quantum brute-force key search.
\end{description}

\subsection{Quantum Security Analysis of 3XCE}
For 3XCE, we design quantum MITM attacks in both Q1 and Q2 models.

\begin{description}
  \item[\textbf{Attacks in Q2 Model}:] By treating 3XCE as a special case of generalized 2kTE, we adapt QCF and Grover-based strategies. The Grover-based variant achieves $O(2^{(\kappa+n)/2})$ time complexity with $O(2^n)$ QRAM, outperforming classical MITM ($O(2^{\kappa+n})$) and Grover-based brute-force attack ($O(2^{(\kappa+2n)/2})$).
  \item[\textbf{Attacks in Q1 Model}:] A novel quantum MITM attack without QRAM or classical preprocessing achieves $O(2^{(\kappa+n)/2})$ time complexity, providing a quadratic speedup over classical counterparts. This attack leverages the linear property of XOR operations to sieve candidate keys efficiently.
\end{description}

\subsection{Generalization to Quantum SITM Attack}
We introduce a quantum SITM framework, extending MITM to constructions of the form $ELE=E^2\circ L\circ E^1$. Key advancements include:
\begin{description}
  \item[\textbf{General Framework}:] A distinguisher-based approach to filter candidate keys using intermediate state properties (e.g., XOR differences or linear involutions).
  \item[\textbf{Three Applications}:] We show three applications.
  \noindent\begin{description}
          \item[3XCE:] When applied to 3XCE, the framework exploits the linearity of XOR operations. By checking for consistent differences in intermediate states ($a_i\oplus a_j=b_i\oplus b_j$), the distinguisher efficiently sieves valid key candidates, achieving a time complexity of $O(2^{(\kappa+n)/2})$ in the Q1 model without QRAM.
          \item[KARC:] For KARC-like constructions with linear involutions, the distinguisher detects fixed differences in intermediate states (e.g., $b_i\oplus\mathcal{R}(a_i)$ is invariant), enabling efficient key recovery with $O(2^{(\kappa+n)/2})$ complexity.
          \item[Involution-Based Constructions:] The attack is applied to a general construction $Enc=E^2\circ L\circ E^1$ where the middle layer $L$ is an involution or satisfies $L\circ P\circ L=Id$ ($P$ is a fixed operation). The key consideration is combining mirror slide attacks with quantum queries to reduce complexity in Q2 model.
        \end{description}
\end{description}
\subsection{Implications and Open Challenges}
\begin{description}
  \item[\textbf{QRAM-Dependent Security}:] According to our quantum MITM attacks, the QRAM size will affect the conclusion of security analysis. If QRAM size is not allowed to exceed the time complexity, the attacks on 2kTE would have higher time complexity than Grover-based brute-force attack on underlying block cipher. This indicates that 2kTE may enhance the quantum security of underlying block ciphers. However, 2kTE cannot enhance the classical security since it can be broken in time $O(2^\kappa)$ by classical MITM attack.
  \item[\textbf{Broader Cryptanalysis}:] Firstly, although the general framework of quantum SITM attack is proposed for $ELE=E^2\circ L\circ E^1$, it is also applicable for the case that $E^1=Id$ or $E^2=Id$. Perhaps this will be more useful in quantum cryptanalysis. Secondly, the quantum SITM framework offers a versatile tool for attacking KLE constructions. It is also applicable in quantum cryptanalysis of block ciphers or iterative structures. Finally, while this paper demonstrates the combination of quantum SITM attacks with mirror slide attacks, future research could explore integrating it with other techniques such as biclique attackp or linear cryptanalysis.
  \item[\textbf{Practical Considerations}:] As QRAM implementations remain hypothetical, optimizing memory usage in real-world quantum attacks remains a critical challenge. While some attacks in this paper assume Q2 model, they highlight the need for KLE designs that explicitly account for quantum adversaries, particularly in quantum-resistant cryptography standards.
\end{description}

\subsection{Conclusion}
This study underscores the fragility of traditional KLE techniques under quantum computing, demonstrating that naive extensions of classical ciphers may fail to withstand quantum MITM and SITM attacks. The proposed quantum SITM framework opens new avenues for analyzing complex cryptographic constructions, urging the community to prioritize quantum-secure KLE designs with rigorous provable security against both classical and quantum adversaries.



\end{document}